\documentclass[12pt]{article}
\usepackage{geometry}                
\usepackage[utf8]{inputenc}
\usepackage{graphicx}
\usepackage{epsfig}
\usepackage{color}
\usepackage{subfigure}
\usepackage{graphics}

\usepackage{amssymb}
\usepackage{amsthm}
\usepackage{amsmath}
\usepackage{amscd}
\usepackage{epstopdf}

\usepackage{algorithm2e}
\usepackage{algcompatible}
\usepackage{fullpage}

\newtheorem{theorem}{Theorem}
\newtheorem{lemma}[theorem]{Lemma}
\newtheorem{proposition}[theorem]{Proposition}
\newtheorem{corollary}[theorem]{Corollary}

\newtheorem{definition}[theorem]{Definition}

\newtheorem{invariant}[theorem]{Invariant}

\usepackage{pgf,tikz}
\usetikzlibrary{arrows}
\definecolor{qqqqff}{rgb}{0,0,1}

\usetikzlibrary{positioning}
\usetikzlibrary{arrows}

\tikzset{plein/.style={very thick}}
\tikzset{pointille/.style={thick,dashed}}
\tikzset{fleche/.style={->, >=open triangle 45}}
\tikzset{revfleche/.style={<-, >=open triangle 45}}
\tikzset{arc/.style={->, >=open triangle 45}}

\title{Read networks and k-laminar graphs}

\author{Finn Völkel \thanks{IRIF, CNRS and Univ. Paris Diderot, France}%
\and%
 Eric Bapteste \thanks{Team Adaptation, Integration, Reticulation, Evolution Lab. CNRS and IBPS, Univ. Pierre et Marie Curie, Paris, France}%
\and%
 Michel Habib $^{*}$ \thanks{GANG Project, Inria Paris, France}%
\and%
 Philippe Lopez$^{\dag}$ %
\and%
 Chloe Vigliotti$^{\dag}$\thanks{MECADEV, CNRS and Museum National d'Histoire Naturelle, Paris}}
\date{\today}

\begin{document}
\maketitle

\begin{abstract}
In this paper we introduce k-laminar graphs a new class of graphs  which extends the idea of Asteroidal triple free graphs. Indeed a graph is k-laminar if it admits a diametral path that is k-dominating. This bio-inspired class of graphs  was motivated by a biological application dealing with sequence similarity networks of reads (called hereafter read networks for short). We briefly develop the context of the biological application in which this graph class appeared and then
we  consider the relationships of this new graph class among  known graph classes and then we study its recognition problem. For the recognition of k-laminar graphs, we develop polynomial  algorithms when  k is fixed. For k=1, our algorithm improves a Deogun and Krastch's algorithm (1999). We finish by an NP-completeness result when k is unbounded.
\end{abstract}


\section{Introduction and biological motivation\label{introduction}}

Roughly speaking a k-laminar graph has a spine and all others vertices are closed to the spine (a more formal definition will be  given in the next section).
The definition of this graph class was motivated by its appearance in reads similarity networks of genomics or metagenomics data \cite{BHBH15} see Figure \ref{reads}. In sequence similarity networks, vertices are biological sequences (either DNA or protein sequences) and two vertices are adjacent if the corresponding sequences are similar, meaning that the pair shows a high enough BLAST score \cite{blast90} and matches over more than 90\% of the longest sequence. Here, sequences come from a metagenomic project and are usually called reads. Basically, reads are raw sequences that come off a sequencing machine, they are random DNA fragments, roughly 300 characters long, coming from the various microbial genomes that are found in a given environment. In our multidisciplinary approach \cite{EHHLV14} we wonder if two species of lizards can be distinguished by the analysis of read networks sequenced from the microbial DNA (microbiome) present in their gastro enteric tract. These networks are useful to biologists because, in addition to allowing the visualization of the genetic diversity that is found in the microbes of a given environment, they offer an alternative to more classical approaches, like the building of contigs\footnote{a contig is a simple path  in the approach based on the de Bruijn graph for assembling reads \cite{PLYC2012}.}, where one tries to rebuild the original genomic sequences of each organism out of the fragments, after a step of binning. The step of binning is a process which clusters contigs or reads, generally based on their composition, and tries to assign them to Operational Taxonomic Units (OTUs - which is the most commonly used microbial diversity unit)\cite{Saw}. Sequence similarity networks are indeed a relaxation of de Bruijn graphs \cite{CPT11}, which are commonly used to build contigs, since they are undirected and, more importantly, since two sequences are adjacent if they show a high enough, but not necessarily exact, similarity. In particular, they allow for the quantification of the genetic diversity of an ensemble of reads. For example of such networks, see Figure \ref{reads}.

When a subset of reads covers a contiguous part of a genome (or parts of the genomes which have the same origin (common ancestor) also called homologous parts), they assemble into a k-laminar graph in sequence similarity networks, thus defining a singular genomic context (e.g. a suite of genes) on which biologists can measure the genetic diversity of the community. However, some genetic sequences, like repeats or transposases\footnote{transposase is a self-replicating enzyme that can insert itself in various part of genome, and is thus found in a variety of genomic contexts.}, can be found in more than one genomic context (i.e. when copies of the same transposase are inserted in multiple distinct locations of a genome), effectively linking together k-laminar graphs in sequence similarity networks. Building contigs out of sequences from such connected components is an especially difficult task.

To sum it up, sequence similarity networks of reads are thus composed of k-laminar parts, corresponding to singular regions of the genomes of a given environment, joined together by groups of vertices corresponding to repetitions in the genomes of a given environment. Identifying k-laminar parts in such networks, and eventually achieving a k-laminar decomposition, is thus of major interest to biologists.

\begin{figure}[ht]
   \begin{minipage}[c]{.46\linewidth}
         \includegraphics[scale=0.25]{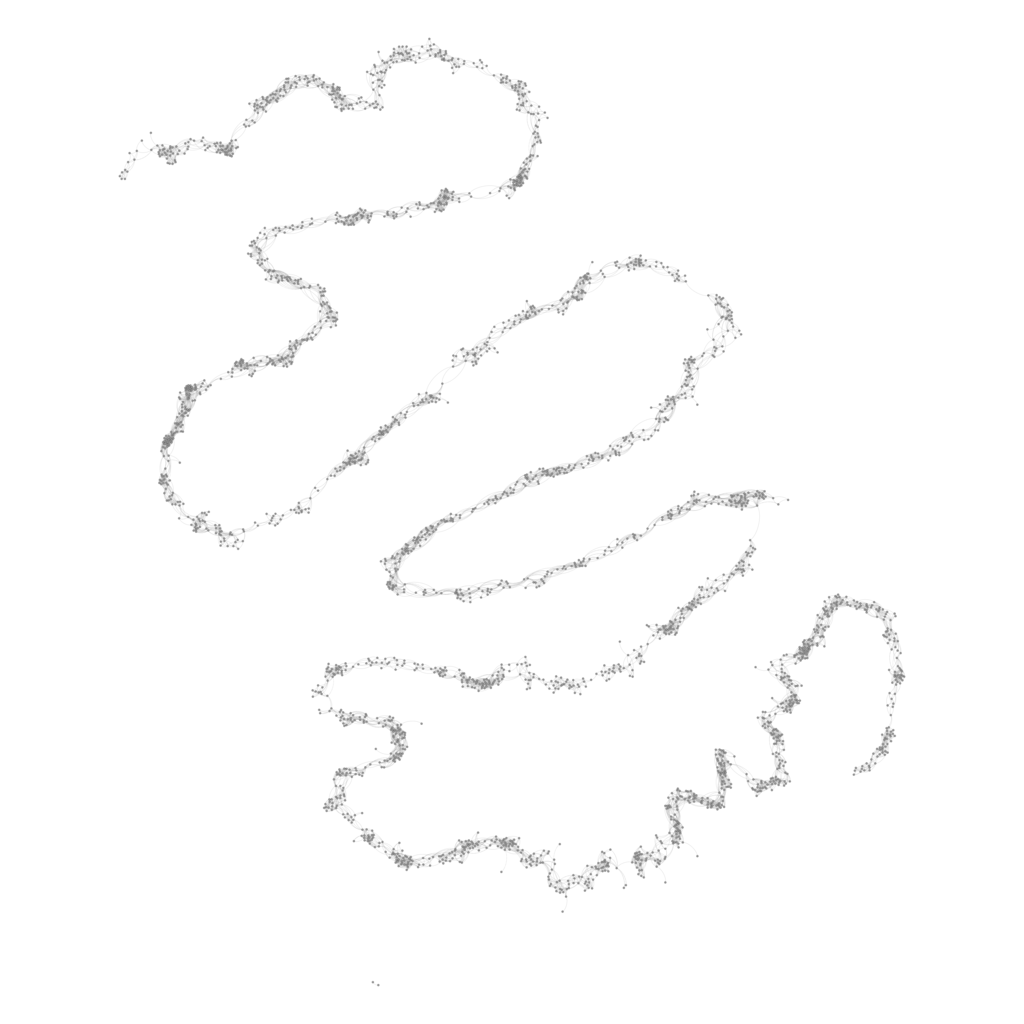}
   \end{minipage} \hfill
   \begin{minipage}[c]{.46\linewidth}
      \includegraphics[scale=0.25]{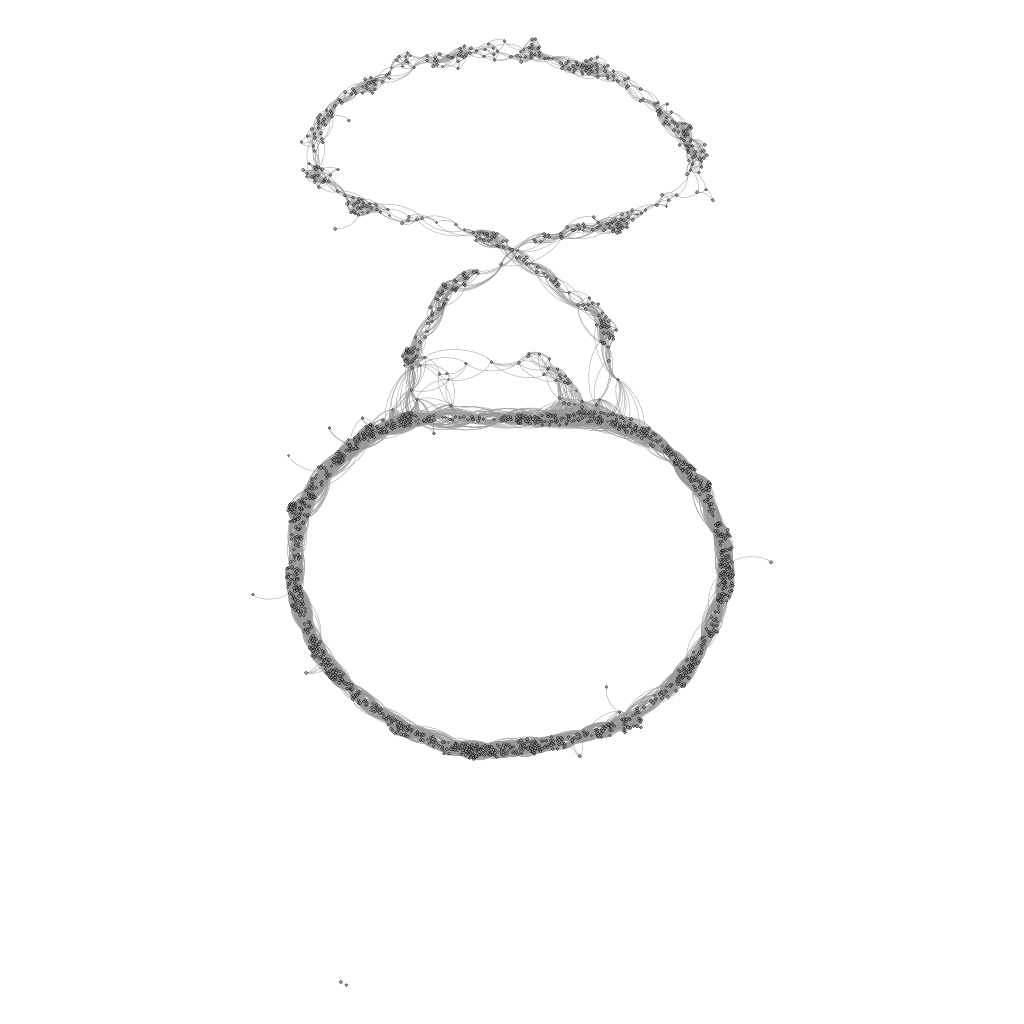}
   \end{minipage}
 \caption{Two read graphs: the left one
   is  a 4-laminar graph, the right one  contains big cycles but can be decomposed into 4-laminars parts. Data used here comes from our project \cite{EHHLV14}.}\label{reads} 
 \end{figure}

\section{k-laminar graphs}
The graphs considered here are finite, loopless and undirected.
For a  connected graph $G$, with vertex set $V(G)$ and edge set $E(G)$,  we denote by $d(x,y)$ for $x, y \in V(G)$ the distance between two vertices, i.e., the length of a shortest path joining $x$ to $y$ in $G$.  We will use also the notion of eccentricity of a vertex  $x \in V(G)$: 
$ecc(x)= max_{y\neq x, y \in V(G)} d(x, y)$ and so the diameter is $diam(G)=max_{x \in V(G)} ecc(x)$, similarly the radius is defined as $radius(G)=min_{x \in V(G)} ecc(x)$. Furthermore let us denote by $MaxEcc(G)$ the set of all vertices of maximum eccentricity. When there is no ambiguity for a graph $G$ we will denote by $n, m$ respectively $|V(G)|, |E(G)|$.

We extend this notion to the distance of vertex to  a path, namely $d(x,\mu)$, for some path $\mu$, is the smallest distance from $x$ to some vertex on $\mu$.  $N(x)$ will be the standard  neighborhood of a vertex and we use also the notation $N[x]=N(x)\cup \{x\}$ for the closed neighborhood.

Similarly $N^k(x)$ the k-neighborhood, i.e. the vertices with distance equal to k from $x$ or more formally $N^k(x) := \{y\ |\ d(x,y) = k \}$. We denote by $N^k[x]$ all vertices with distance less or equal to k, i.e. $N^k[x] := \{y\ |\ d(x,y) \leq k \}$ called the closed k-neighborhood.
When $G$ is connected,  the maximum length of a path is called the diameter of $G$ and denoted by $diam(G)$.

In this section we recall some standard definitions on graphs and introduce the notion of laminar graphs and the practical motivations of such a definition.

\begin{definition}
 	An asteroidal triple (AT) is a triple of vertices such that each pair is joined by a path that avoids the neighborhood of the third.
\end{definition}

An AT-free graph is a graph that does not contain any AT.  Intuitively if a graph does not contain any AT, then it cannot "expand" in more than 2 directions. The following definition  of laminar graphs introduced here, is to generalize this intuitive notion of linearity. 

\begin{definition}
	A path $\mu$  of a graph $G$ is a diametral path if the length of $\mu$ is equal to $diam(G)$. Furthermore for every fixed integer $k$, a  path $\mu$ in a graph $G$ is called a $k$-dominating path if  $\forall x \in V(G)$   we have $d(x, \mu) \leq k$.

\end{definition}

\begin{definition}\label{klaminar}

A graph $G$ is called  $k$-laminar (resp. strongly $k$-laminar) if   $G$ has  a $k$-dominating  diametral path (resp. if every diametral path is  $k$-dominating).
\end{definition}

\begin{proposition}\cite{corneil1}
AT-free graphs are 1-laminar. 
\end{proposition}

\begin{proof}
Corneil, Olariu and Stewart proved that AT-free graphs contain a dominating pair that achieves the diameter. Hence, AT-free graphs are 1-laminar.
\end{proof}

\begin{definition}
	A comparability graph is a graph $G$ whose edge set $E(G)$  can be  transitively  oriented. A cocomparability graph is simply the complement of a  comparability graph.
\end{definition}

It is well known that cocomparability graphs are AT-free \cite{GMT84}. Therefore also cocomparability graphs and interval graphs are 1-laminar.

 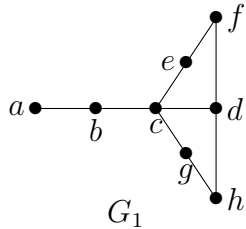
\begin{figure}[ht!] 

    \begin{center}
      \begin{tikzpicture}[scale=0.4, pil/.style={->,thick,shorten <=2pt,shorten >=2pt,>=triangle 60}]

        \coordinate(A) at (1,0);%
        \coordinate(B) at (3,0);%
        \coordinate(C) at (5,0);%
        \coordinate(D) at (7,0);%
	\coordinate(E) at (6,1.5);%
        \coordinate(F) at (7,3);%
        \coordinate(G) at (6,-1.5);%
        \coordinate(H) at (7,-3);%

        \draw(A)node[left]{$a$} node{$\bullet$};%
        \draw(B)node[below]{$b$} node{$\bullet$};%
        \draw(C)node[below]{$c$} node{$\bullet$};%
        \draw(D)node[right]{$d$} node{$\bullet$};%
	 \draw(E)node[left]{$e$} node{$\bullet$};%
        \draw(F)node[right]{$f$} node{$\bullet$};%
        \draw(G)node[below]{$g$} node{$\bullet$};%
        \draw(H)node[right]{$h$} node{$\bullet$};%

        \draw(A)--(B)--(C)--(D);%
\draw(C)--(E)--(F)--(D);
\draw(C)--(G)--(H)--(D);

        \node at (4,-3.5) [scale=1] { $G_1$};%

      \end{tikzpicture}
    \end{center}
    \caption{ $\mu=[a, b, c, d, h]$ is a dominating diametral path of $G_1$, and $(a, f, h)$ is an AT.}\label{ATnonstrongly}        
    
     \end{figure}

As shown by the graph $G_1$ of Figure \ref{ATnonstrongly}, not all 1-laminar graphs are AT-free graphs. Thus AT-free graphs $\subsetneq$ 1-laminar. Furthermore AT-free  are not always strongly 1-laminar as can be seen with the graph $G_2$ of Figure \ref{nonAT}.
To complete the picture, the class of AT-free graphs overlaps the class of 1-strongly laminar as can be seen with the graph $G_3$ of  Figure \ref{stronglynonAT}.
 
 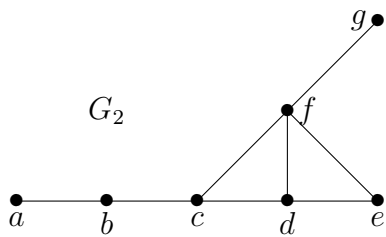
\begin{figure}[ht!]
    \begin{center}
      \begin{tikzpicture}[scale=0.4, pil/.style={->,thick,shorten <=2pt,shorten >=2pt,>=triangle 60}]

        \coordinate(A) at (0,0);%
        \coordinate(B) at (3,0);%
        \coordinate(C) at (6,0);%
        \coordinate(D) at (9,0);%
	\coordinate(E) at (12,0);%
        \coordinate(F) at (9,3);%
        \coordinate(G) at (12,6);%

\node at (3, 3) [scale=1] { $G_2$};

        \draw(A)node[below]{$a$} node{$\bullet$};%
        \draw(B)node[below]{$b$} node{$\bullet$};%
        \draw(C)node[below]{$c$} node{$\bullet$};%
        \draw(D)node[below]{$d$} node{$\bullet$};%
	 \draw(E)node[below]{$e$} node{$\bullet$};%
        \draw(F)node[right]{$f$} node{$\bullet$};%
        \draw(G)node[left]{$g$} node{$\bullet$};%

        \draw(A)--(B)--(C)--(D)--(E);%
\draw(C)--(F)--(E);
\draw (G)--(F)--(D);

      \end{tikzpicture}
    \end{center}
    \caption{$G_2$ is  AT-free  but $\mu=[a, b, c, d, e]$ is a non dominating  diametral path of $G_2$.}\label{nonAT}        
        \end{figure}
 \begin{figure}[ht!]
    \begin{center}
      \begin{tikzpicture}[scale=0.4, pil/.style={->,thick,shorten <=2pt,shorten >=2pt,>=triangle 60}]

        \coordinate(A) at (0,0);%
        \coordinate(B) at (3,0);%
        \coordinate(C) at (6,0);%
        \coordinate(D) at (9,0);%
	\coordinate(E) at (12,0);%
        \coordinate(F) at (0,4);%
        \coordinate(G) at (3,4);%
        \coordinate(H) at (6,4);%
        \coordinate(I) at (9,4);%
        \coordinate(J) at (12,4);%

\node at (-3,2) [scale=1] { $G_3$};

        \draw(A)node[below]{$a$} node{$\bullet$};%
        \draw(B)node[below]{$b$} node{$\bullet$};%
        \draw(C)node[below]{$c$} node{$\bullet$};%
        \draw(D)node[below]{$d$} node{$\bullet$};%
	 \draw(E)node[below]{$e$} node{$\bullet$};%
        \draw(F)node[above]{$f$} node{$\bullet$};%
        \draw(G)node[above]{$g$} node{$\bullet$};%
        \draw(H)node[above]{$h$} node{$\bullet$};%
        \draw(I)node[above]{$i$} node{$\bullet$};%
        \draw(J)node[above]{$j$} node{$\bullet$};%

        \draw(A)--(B)--(C)--(D)--(E);%
\draw(F)--(G)--(H)--(I)--(J);
\draw(A)--(F);
\draw(E)--(J);
\draw (A)--(G)--(C);
\draw(G)--(B)--(H);
\draw (I)--(E);
\draw(D)--(J);

      \end{tikzpicture}
    \end{center}
    \caption{$G_3$ is not AT-free, (g, i, d) is an AT, but $G_3$ is 1-strongly laminar.}\label{stronglynonAT}        
        \end{figure}
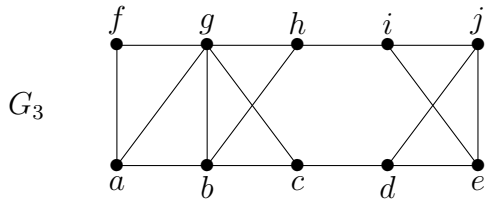

The smallest k such that a graph is k-laminar is called the laminar index of $G$ and denoted by $Laminar(G)$. This invariant is well defined since obviously $Laminar(G) \leq diam(G)$ and furthermore if a center of the graph belongs to a diametral path: $Laminar(G)\leq radius(G)$.
This paper is devoted to the study of  (strongly) k-laminar graphs, their structure but also the existence of polynomial recognition algorithms.
Since it is well known that a graph may have an exponential number of diametral paths (see for example the graph in Figure \ref{fig:reduction}), at first glance we can only state that the recognition problem of strongly k-laminar paths is in \textbf{co-NP}.
In  \cite{kratsch1}, Deogun and Kratsch introduced  a very similar graph class, namely  the diametral path graphs.

\begin{definition}\cite{kratsch1}
 A graph G is called a diametral path graph if every connected induced subgraph H of G has a dominating diametral path or equivalently $H$ is a 1-laminar graph using definition \ref{klaminar}.
\end{definition}
It is not hard to see that all diametral path graphs are 1-laminar.
But 1-laminar graphs strictly contain diametral path graphs, as can be seen with the graph $G_1 \setminus\{d\}$ which is no 1-laminar. Therefore $G_1 \in$ 1-laminar graphs $\setminus$  diametral path graphs.

Using  a polynomial time algorithm \cite{kratsch1} for testing if a graph has a dominating diametral path, we know that the recognition of 1-laminar graphs is polynomial.

Moreover Deogun and Kratsch were able to prove that diametral path graphs that are trees or chordal graphs have simple forbidden subgraphs (polytime-recognizable).
But to our knowledge it is still an open question whether diametral path graphs can be recognized in polynomial time. 

\begin{figure}[!h]
\centering
\begin{tikzpicture}[scale = 0.25,pil/.style={->,thick,shorten <=2pt,shorten >=2pt,>=triangle 60}]
\tikzstyle{fontbf} = [text width=5cm,text centered,font=\bf]
 \draw [ultra thick,rounded corners = 5pt] (-4,-4) rectangle (50,42);

 \draw [ultra thick,rounded corners = 5pt] (0,0) rectangle (48,41);
 \draw [ultra thick,rounded corners = 5pt] (2,15) rectangle (46,39);
 \draw [ultra thick,rounded corners = 5pt] (4,2) rectangle (28,27);
 \draw [ultra thick,rounded corners = 5pt] (22,19) rectangle (44,35);
 
 \node at (37,-1) [fontbf] {\dots};
 \node at (37, -3) [fontbf] {k-laminar};
 \node at (37,2) [fontbf] {1-laminar};
 \node at (33,37) [fontbf] {Diametral Path Graphs};
 \node at (13.6,4) [fontbf] {1-strongly laminar};
 \node at (36,33) [fontbf] {AT-free};
 
\begin{scope}[shift={(33,9)}]
	\coordinate(A) at (1,0);%
    \coordinate(B) at (3,0);%
    \coordinate(C) at (5,0);%
    \coordinate(D) at (7,0);%
	\coordinate(E) at (6,1.5);%
    \coordinate(F) at (7,3);%
    \coordinate(G) at (6,-1.5);%
    \coordinate(H) at (7,-3);%

    \draw(A)node[left]{$a$} node{$\bullet$};%
    \draw(B)node[below]{$b$} node{$\bullet$};%
    \draw(C)node[below]{$c$} node{$\bullet$};%
    \draw(D)node[right]{$d$} node{$\bullet$};%
	\draw(E)node[left]{$e$} node{$\bullet$};%
    \draw(F)node[right]{$f$} node{$\bullet$};%
    \draw(G)node[below]{$g$} node{$\bullet$};%
 	\draw(H)node[right]{$h$} node{$\bullet$};%

    \draw(A)--(B)--(C)--(D);%
	\draw(C)--(E)--(F)--(D);%
	\draw(C)--(G)--(H)--(D);%
	
	\node at (4,-3.5) [scale=1] { $G_1$};%
\end{scope}

\begin{scope}[shift={(30,23)}]
	\coordinate(A) at (0,0);%
    \coordinate(B) at (3,0);%
    \coordinate(C) at (6,0);%
    \coordinate(D) at (9,0);%
	\coordinate(E) at (12,0);%
    \coordinate(F) at (9,3);%
    \coordinate(G) at (12,6);%
        
    \draw(A)node[below]{$a$} node{$\bullet$};%
    \draw(B)node[below]{$b$} node{$\bullet$};%
    \draw(C)node[below]{$c$} node{$\bullet$};%
    \draw(D)node[below]{$d$} node{$\bullet$};%
	\draw(E)node[below]{$e$} node{$\bullet$};%
    \draw(F)node[right]{$f$} node{$\bullet$};%
    \draw(G)node[left]{$g$} node{$\bullet$};%
        
    \draw(A)--(B)--(C)--(D)--(E);%
	\draw(C)--(F)--(E);%
	\draw (G)--(F)--(D);%
	
	\node at (3, 3) [scale=1] { $G_2$};
\end{scope}

\begin{scope}[shift={(7,17)}]
	\coordinate(A) at (0,0);%
    \coordinate(B) at (3,0);%
    \coordinate(C) at (6,0);%
    \coordinate(D) at (9,0);%
	\coordinate(E) at (12,0);%
    \coordinate(F) at (0,4);%
    \coordinate(G) at (3,4);%
    \coordinate(H) at (6,4);%
    \coordinate(I) at (9,4);%
    \coordinate(J) at (12,4);%
	
	\draw(A)node[below]{$a$} node{$\bullet$};%
    \draw(B)node[below]{$b$} node{$\bullet$};%
    \draw(C)node[below]{$c$} node{$\bullet$};%
    \draw(D)node[below]{$d$} node{$\bullet$};%
	\draw(E)node[below]{$e$} node{$\bullet$};%
    \draw(F)node[above]{$f$} node{$\bullet$};%
    \draw(G)node[above]{$g$} node{$\bullet$};%
    \draw(H)node[above]{$h$} node{$\bullet$};%
    \draw(I)node[above]{$i$} node{$\bullet$};%
    \draw(J)node[above]{$j$} node{$\bullet$};%

    \draw(A)--(B)--(C)--(D)--(E);%
	\draw(F)--(G)--(H)--(I)--(J);
	\draw(A)--(F);
	\draw(E)--(J);
	\draw (A)--(G)--(C);
	\draw(G)--(B)--(H);
	\draw (I)--(E);
	\draw(D)--(J);

	\node at (6,8) [scale=1] { $G_3$};
\end{scope}

\begin{scope}[shift={(25,23)}]
	\node at (0,0) [scale=1] {$G_2 \setminus d$};
\end{scope}

\begin{scope}[scale=0.75,shift={(11,41)}]
	\coordinate(A) at (0,0);%
    \coordinate(B) at (3,0);%
    \coordinate(C) at (9,0);%
    \coordinate(D) at (12,0);%
    \coordinate(E) at (6,3);%
    \coordinate(F) at (6,-3);%
    \coordinate(G) at (6,6);%
    
    \draw(A)node[below]{$a$} node{$\bullet$};%
    \draw(B)node[below]{$b$} node{$\bullet$};%
    \draw(C)node[below]{$c$} node{$\bullet$};%
    \draw(D)node[below]{$d$} node{$\bullet$};%
	\draw(E)node[left]{$e$} node{$\bullet$};%
    \draw(F)node[left]{$f$} node{$\bullet$};%
    \draw(G)node[above]{$g$} node{$\bullet$};%
    
    \draw(A)--(B)--(E)--(G);%
    \draw(B)--(F)--(C)--(D);%
    \draw(E)--(C);%
    
	\node at (2,3) [scale=1] {$G_4$};

\end{scope}

\begin{scope}[scale=0.75,shift={(18,12)}]
	\coordinate(A) at (0,0);%
    \coordinate(B) at (2,0);%
    \coordinate(C) at (4,2);%
    \coordinate(D) at (4,4);%
    \coordinate(E) at (2,6);%
    \coordinate(F) at (0,6);%
    \coordinate(G) at (-2,4);%
    \coordinate(H) at (-2,2);%
    \coordinate(I) at (1,-2);%
    \coordinate(J) at (3,-2);%
    \coordinate(K) at (5,-2);%
    \coordinate(L) at (7,-2);%
    \coordinate(M) at (-1,-2);%
    \coordinate(N) at (-3,-2);%
    \coordinate(O) at (-5,-2);%
    
    \draw(A)node[left]{$a$} node{$\bullet$};%
    \draw(B)node[right]{$b$} node{$\bullet$};%
    \draw(C)node[right]{$c$} node{$\bullet$};%
    \draw(D)node[right]{$d$} node{$\bullet$};%
	\draw(E)node[right]{$e$} node{$\bullet$};%
    \draw(F)node[left]{$f$} node{$\bullet$};%
    \draw(G)node[left]{$g$} node{$\bullet$};%
    \draw(H)node[left]{$h$} node{$\bullet$};%
    \draw(I)node[below]{$i$} node{$\bullet$};%
    \draw(J)node[below]{$j$} node{$\bullet$};%
    \draw(K)node[below]{$k$} node{$\bullet$};%
	\draw(L)node[below]{$l$} node{$\bullet$};%
    \draw(M)node[below]{$m$} node{$\bullet$};%
    \draw(N)node[below]{$n$} node{$\bullet$};%
    \draw(O)node[below]{$o$} node{$\bullet$};%
    
    \draw(A)--(B)--(C)--(D)--(E)--(F)--(G)--(H)--(A);%
    \draw(O)--(N)--(M)--(I)--(J)--(K)--(L);%
    \draw(A)--(I)--(B);%
    \draw(H)--(I)--(C);%
    \draw(G)--(I)--(D);%
    \draw(F)--(I)--(E);%
    
	\node at (-7,2) [scale=1] {$G_5$};

\end{scope}

\end{tikzpicture}   
\caption{Relationships among  the main graph  classes considered so far, including  that  \textbf{AT-free} $\subset$  \textbf{Diametral Path Graphs} as first noticed in \cite{Khab}.} \label{Venndiagrams}
\end{figure}
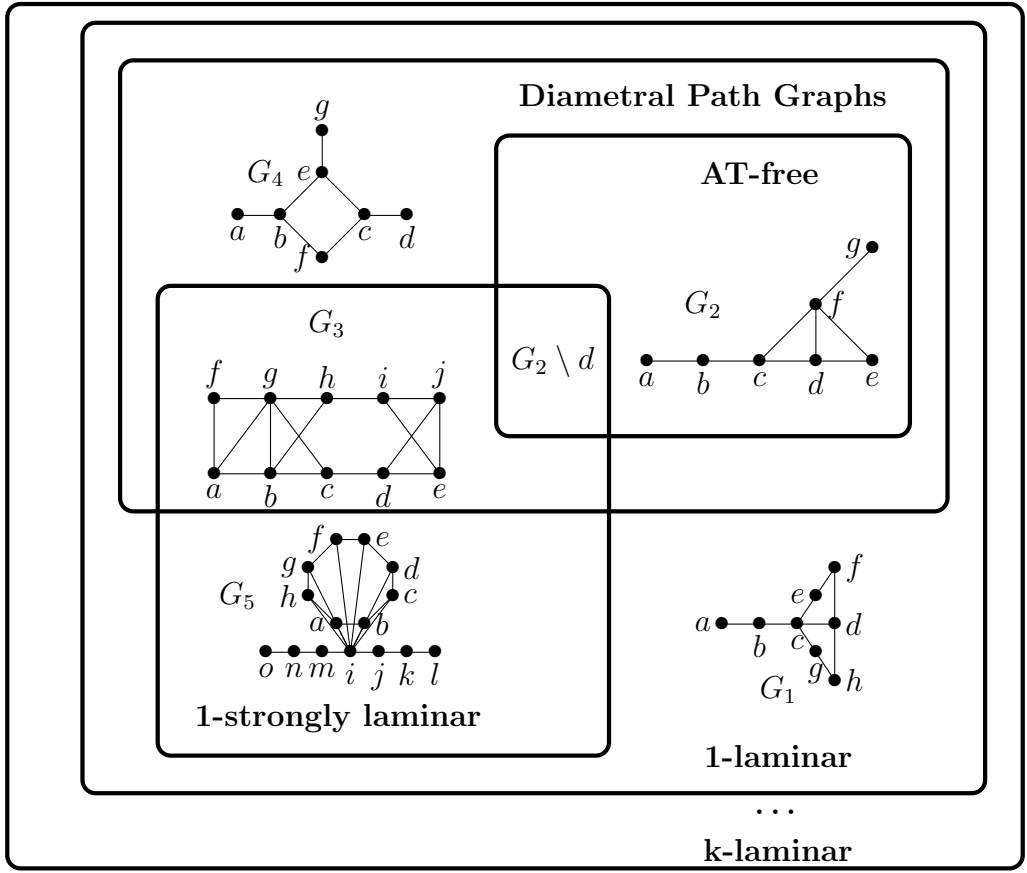

The remaining part of the paper is organized as follows:

In section \ref{polynomial} we show that the recognition of strongly laminar graphs is polynomial, such as the recognition of k-laminar graphs when k is fixed. To this aim we improve an algorithm from \cite{kratsch1} to recognize 1-laminar graphs and we generalize it for every fixed k.

In section \ref{reduction} we present strong evidence that it is intractable to find the laminar index. In fact we present a reduction which proves that recognizing if a  graph $G$ is
$k$-laminar  is NP-complete, for a given range of  $k$ values in  $O(\sqrt{|V(G)|})$.

\section{Polynomial algorithms}\label{polynomial}

The main contribution of this section is that we present a polynomial recognition algorithms  for any fixed k for k-laminar  (resp. strongly k-laminar) graphs. Let us  begin with the strongly case.

\subsection{Strongly k-laminar graphs}

First we need an easy  lemma.

\begin{lemma}\label{subgraph}
Let $x \in MaxEcc(G)$ and $H$ be an induced subgraph of $G$ containing $x$. If $ecc_H(x)=diam(G)$ then there exists a shortest path $\mu=[x,y]$ in $H$ and $G$ of size $diam(G)$.
\end{lemma}

\begin{proof}
We notice that $ecc_G(x) \leq ecc_H(x)$. In case of equality it yields that there exists a path $\mu=[x,y]$  in $H$ of length $diam(G)$.
$\mu$ is still a shortest path with no shortcut in $G$, unless $x \in MaxEcc(G)$.
\end{proof}

\begin{theorem}\label{kstrongly}
	The recognition of strongly k-laminar graphs can be done in $O(|MaxEcc(G)| nm)$ bounded by $O(n^2m)$ for every fixed $k$.
\end{theorem}
\begin{proof}
If a graph is not strongly k-laminar then there exists some diametral path that does not pass through the k-neighborhood of some vertex $x$. It suffices therefore to verify that every diametral path passes through $N^k[x]$ $\forall x \in V$. This can easily be done by recalculating the distance matrix  in $G \setminus N^k[x]$ for every $x$. We know that $diam(G \setminus N^k[x]) \geq diam(G)$.

If for some vertex $x$ $d_{G \setminus N^k[x]}(a, b)  = d_G(a, b)= diam(G)$, using lemma \ref{subgraph} we know
there exsit some path  $\mu$ in $G$ 
which is a diametral path that does not pass through 
$N^k[x]$ and therefore the strongly laminar condition is not satisfied.

We need for every vertex $x$ to compute $G'=G \setminus N^k[x]$ which can be done $O(|V(G|+|E(G)|)$ using a BFS. But then we must compute the eccentricity of all $MaxEcc(G)$ vertices  in $G'$ which can be done in a naive way by processing $|MaxEcc(G)|$ BFS's in $O( |V(G')|.|E(G')|)$.  

Therefore for each $k$ this can be done in $O(|MaxEcc(G)| nm)$, i.e., in $O(n^2m)$.
\end{proof}

As an immediate consequence:

\begin{corollary}
The computation of the smallest k for which a graph $G$ is k-strongly laminar is polynomial.
\end{corollary}

\begin{proof}
Since $1\leq k \leq n-1$ and we can use a dichotomic process of the above algorithm, which yields a complexity of
$O(log(n) n^2 m)$. 
\end{proof}


\subsection{1-laminar graphs}

\vspace{0.5cm}

Let us now describe an
 improved variation of the $O(n^3m)$ Deogun and Kratsch's algorithm \cite{kratsch1}, searching for the existence of a dominating diametral path in $O(n^2m)$. 

As a preprocessing, we can compute $ecc(x)$ for every vertex $x$ of $G$. Afterwards $\forall s \in MaxEcc(G)$, we process a BFS and let us denote by  $T_s$ the associated BFS-tree.  $L_i$ represent the different layers of the BFS-tree, i.e.  by convention $L_0=\{s\}$ and $L_i$ is equal to the i-th neighborhood of $s$. Then $\forall v \in V(G)$, let us denote by $Level_s(v)$ its level in $T_s$.
We can also preprocess in linear time :
$\forall v \in V(G)$ and   $\forall i$ such that  $Level_s(v)$-$1 \leq i \leq Level_s(v)+1$  we compute  $ N_i(v)=N(v) \cap L_i$. 

Then we can use for every vertex $s \in MaxEcc(G)$ the following modified BFS, which is in fact  a partial BFS since only the vertices that can be part of a dominating diametral path are explored.

\begin{figure}[ht!]
\begin{algorithm}[H]
\textbf{Dominating-Diameter(G,s):} 

\KwData{a graph $G=(V,E)$ and a start vertex $s \in MaxEcc(G)$\;}
\KwResult{YES / NO  $G$ has a dominating diametral path starting at $s$\;}

Mark FEASIBLE all edges adjacent to $s$.
Initialize \textit{Queue} to  $N(s)$\;
\While{$Queue \neq \emptyset$}{dequeue $v$ from beginning of  \textit{Queue}\;
 $h \leftarrow Level_s(v)$\;

\For{$\forall u \in N_{h-1}(v)$ with $uv$ marked FEASIBLE}{

$A(v) \leftarrow N_h(v) \cup N_h(u)$\;
\If{$h=diam(G)$}{\If {$L_h=A(v)$}
{ \textbf{YES} a dominating diametral path  from $s$ to $v$ has been found \textbf{STOP}}
}
\For{ $\forall w \in N_{h+1}(v)$}{

\If{$L_{h}$=$A(v) \cup N_h(w)$}{

Mark $vw$ as FEASIBLE\;
\If{$w$ is not already in \textit{Queue}}{
			enqueue $w$ to the end of \textit{Queue} }
			
}
}
}		
}
\textbf{NO}  $G$ has no dominating diametral path starting at $s$\;
\vspace{0.5cm}
\caption{ A modified Breadth First Search}
\end{algorithm}
\end{figure}



\begin{theorem}\label{1-lam}
	Algorithm Dominating-Diameter(G,s) computes if a graph $G$ admits a dominating diametral path starting from $s$ in $O(n m)$.
\end{theorem}

\begin{proof}
Any diametral path must go sequentially through the all the layers of $H_i$. Furthermore using the BFS-tree structure any edge
$xy \in V(G)$ satisfies $|Level_s(x)-Level_s(y)| \leq 1$.

In order to prove the modified BFS algorithm we need to prove that for every vertex  $s$ of maximal eccentricity it is enough to check that the following easy  invariant s:

\begin{invariant}  For all i, $1 \leq i \leq Diam(G)$, and for every $v \in L_i(s)$ If  $v \in Queue$, then there exists a path from $s$ to $v$ in $G$,
that dominates the first $i-1$  layers. Moreover  all these dominating paths reach $v$ with an edge marked FEASIBLE. 
\end{invariant}

\textbf{Complexity analysis:}
The preprocessing time, i.e., computing all eccentricities can be done in a naive way by processing $n$ Breadth First searches (BFS) in $O(nm)$.

Let us consider a  BFS search starting at some $s \in MaxEcc(G)$  and its BFS numbering $\tau$ (the visiting ordering of the vertices during the BFS), one can easily sort all the neighborhood lists of all the vertices according to $\tau$ in linear time. Then for every vertex $x \in L_h$, $N_{h-1}(x), N_{h}(x)$ and $N_{h+1}(x)$ can be extracted from $N(x)$ in $O(1)$.
Therefore for each BFS before using the modified BFS, the preprocessing requires $O(n+m)$.
The structure of the modified BFS, i.e., the while loop, is a partial BFS visiting only vertices that can still belong to a dominating path. Let us now consider the inside instructions.

For every edge $uv$  the test $L_{h}$=$A(v) \cup N_h(w)$ can be done by computing $A(v) \cap N_{h}(w)$ in $O(|A(v)|+| N_{h}(w)|)$ since they are encoded as sorted lists and then comparing the sizes
$|A(v) \cap N_{h}(w)|$ and $|L_{h}|$ in $O(1)$. 
  
For every vertex $v \in L_h$, in the whole :
$N_h(v)$ is used  at most $|N_{h-1}(v)| +|N_{h+1}(v)|$ times.

Therefore for all $v$ it is bounded by   $\Sigma_v |N_h(v)| (  |N_{h-1}(v)| +|N_{h+1}(v)|)$.
Bounding  $|N_{h-1}(v)| +|N_{h+1}(v)|$ by $n$ we obtain:
$n.\Sigma_v d(v) \in O(n.m)$
Therefore the overall time complexity of this algorithm  is $O(nm)$.
\end{proof}

\begin{corollary}
The recognition of 1-laminar graphs can be done in $O(|MaxEcc(G)|.nm)$ bounded by  $O(n^2m)$.
\end{corollary}

\begin{proof}
To recognize if a graph is 1-laminar, it is enough to process for every $s \in MaxEcc(G)$ the algorithm Dominating-Diameter(G,s).
Including the preprocessing and the computation of all eccentricities in $G$ in $O(nm)$, the overall time complexity is $O(|MaxEcc(G)|.nm)$ bounded by  $O(n^2m)$.
\end{proof}

This algorithm can be easily adapted to compute a 1-dominating diametral  path and generalized for every fixed integer $k$, and this yields :
\begin{theorem}
The recognition of k-laminar graphs can be done in $O(n^{2k+1})$.
\end{theorem}
For a proof the reader is referred to the Appendix.

\section{NP-completeness \label{reduction}}

In this section we give a reduction from 3SAT to the recognition of $k$-laminar graphs. It is therefore NP-hard to compute $Laminar(G)$. The reader is encouraged to look at figure \ref{fig:reduction} for an better understanding of the reduction. In this section $n$ denotes the number of variables in a satisfiability formula and $N$ the number of vertices in a graph. Capital letters are used to denote vertices and small letters to denote variables.

Given a 3SAT formula $\phi$ made up with $m$ clauses  $C_j$, $1\leq j \leq m$, on $n$ boolean variables $x_i$  $1\leq i \leq n$. 

We construct a graph $G(\phi)$ and we will prove that:
$G(\phi)$ is  $(\frac{n}{2}+1)$-laminar iff $\phi$ is satisfiable.

\vspace{0.5cm}

Let us first detail the construction of $G(\phi)$.
For each  literal $x_i$ (resp. $\overline{x_i}$) we associate  a vertex $X_i$ (resp. $\overline{X_i}$).
 We put an edge between a variable and its negation. Moreover we connect the vertices $X_i$ and $\overline{X_i}$ with $X_{i-1}, \overline{X_{i-1}}, X_{i+1}, \overline{X_{i+1}}$ if existent. 
 We add a pending chain  $V_1, \dots V_n$  to $X_1$ and $\overline{X_1}$. The same is done symmetrically with a pending chain $V_{n+1}, \dots V_{2n}$ attached to $X_n$ and $\overline{X_n}$. 
Up to now we have $2^n$ shortest paths of length $3n+2$ going from $V_1$ to $V_{2n}$.
 Now for every clause $C_i$, $1\leq j \leq m$, we add a vertex $C_i$. Every $C_i$ is connected by a chain of length $\frac{n}{2} + 1$  to every vertex associated to a litteral  that appears in the clause $C_i$. Note that here for sake of simplicity $n$ is supposed to be even,  otherwise we would add a dummy variable.

 Suppose for now that the diametral path starts and ends from the end vertices of the two chains respectively ($V_1$ and $V_{2n})$. Such a diametral path will never pass through $X_i$ and $\overline{X_i}$ because it would either need to use an edge $X_i\overline{X_i}$ or do some detour which would mean that the length of the path is greater than the diameter $3n+2$. 
 

The graph $G(\phi)$ contains exactly $4n +m_{\phi}(\frac{n}{2}+1)= |V(G)| $  vertices, where $m_{\phi}$ is the total number of variables in the clauses $C_j$.


\begin{lemma}\label{necessaire}
$diam(G(\phi)) = 3n+2$ 
\end{lemma}
\begin{proof}

For any pair of clauses : $C_j, C_{j'}$ , $d(C_j, C_{j'}) \leq 2(\frac{n}{2} + 1)+n \leq  3n$

Furthermore: Let $p$ (resp. $q$) be the minimum (resp. maximum)  index of a literal in $C_j$. 

Then $ecc(C_j) =max\{ n-p +n+1, n+1+q \}=max\{ 2n-p +1, n+q+1\} \leq 2n+1.$

We already have seen that : $ecc(V_1) \leq 3n+2$, using a path going only through the $X_i$'s up to $V_{2n}$.
Moreover no $C_j$ can provide a shortcut to this path. Thus $ecc(V_1) = 3n+2 =ecc(V_{2n})$

\end{proof}


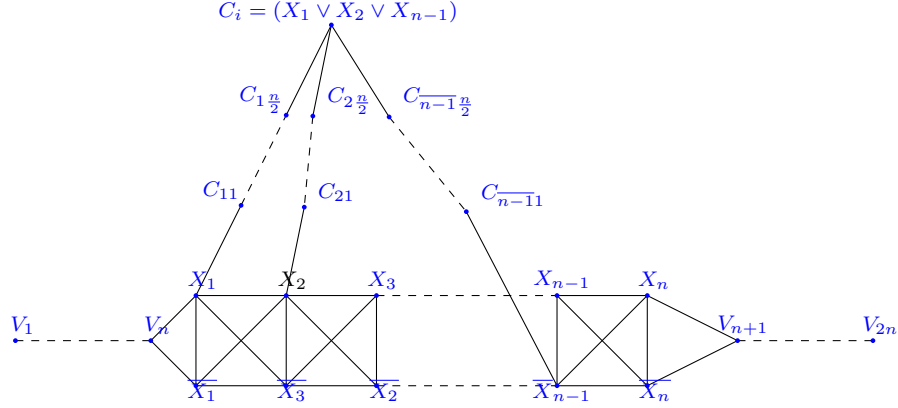
\begin{figure}
\label{fig:reduction}
\hspace*{2.5cm}%
\begin{tikzpicture}[scale=0.6]

\draw (-1,2)-- (0,3);
\draw (-1,2)-- (0,1);
\draw (0,3)-- (4,3);
\draw (0,1)-- (4,1);
\draw (8,3)-- (10,3);
\draw (8,1)-- (10,1);
\draw (10,3)-- (12,2);
\draw (10,1)-- (12,2);
\draw (12,2)-- (15,2)[dashed];
\draw (-4,2)-- (-1,2)[dashed];
\draw (1,5)-- (0,3);
\draw (1,5)-- (2,7)[dashed];
\draw (2,7)-- (3,9);
\draw (3,9)-- (2.59,6.98) ;
\draw (2.59,6.98)-- (2.4,4.96) [dashed];
\draw (2.4,4.96)-- (2,3);
\draw (4.28,6.96)-- (3,9);
\draw (4.28,6.96)-- (5.99,4.86)[dashed];
\draw (5.99,4.86)-- (8,1);
\draw (4,3)-- (8,3) [dashed];
\draw (4,1)-- (8,1) [dashed];
\draw (0,3)-- (2,1);
\draw (0,1)-- (2,3);
\draw (0,3)-- (0,1);
\draw (2,3)-- (2,1);
\draw (4,3)-- (4,1);
\draw (2,1)-- (4,3);
\draw (2,3)-- (4,1);
\draw (8,3)-- (10,1);
\draw (8,1)-- (10,3);
\draw (10,3)-- (10,1);
\draw (8,3)-- (8,1);
\begin{scriptsize}
\fill [color=qqqqff] (0,3) circle (1.5pt);
\draw[color=qqqqff] (0.15,3.34) node {$X_1$};
\fill [color=qqqqff] (0,1) circle (1.5pt);
\draw[color=qqqqff] (0.14,1.33) node[below] {$\overline{X_1}$};
\fill [color=qqqqff] (2,3) circle (1.5pt);
\draw     (2.14,3.34) node {$X_2$};
\fill [color=qqqqff] (4,1) circle (1.5pt);
\draw[color=qqqqff] (4.17,1.33) node[below] {$\overline{X_2}$};
\fill [color=qqqqff] (4,3) circle (1.5pt);
\draw[color=qqqqff] (4.13,3.34) node {$X_3$};
\fill [color=qqqqff] (2,1) circle (1.5pt);
\draw[color=qqqqff] (2.12,1.33) node [below]{$\overline{X_3}$};
\fill [color=qqqqff] (10,3) circle (1.5pt);
\draw[color=qqqqff] (10.15,3.34) node {$X_n$};
\fill [color=qqqqff] (10,1) circle (1.5pt);
\draw[color=qqqqff] (10.17,1.33) node[below] {$\overline{X_n}$};
\fill [color=qqqqff] (8,3) circle (1.5pt);
\draw[color=qqqqff] (8.11,3.34) node {$X_{n-1}$};
\fill [color=qqqqff] (8,1) circle (1.5pt);
\draw[color=qqqqff] (8.11,1.33) node[below] {$\overline{X}_{n-1}$};
\fill [color=qqqqff] (-1,2) circle (1.5pt);
\draw[color=qqqqff] (-0.85,2.34) node {$V_n$};
\fill [color=qqqqff] (12,2) circle (1.5pt);
\draw[color=qqqqff] (12.12,2.34) node {$V_{n+1}$};
\fill [color=qqqqff] (-4,2) circle (1.5pt);
\draw[color=qqqqff] (-3.84,2.34) node {$V_1$};
\fill [color=qqqqff] (15,2) circle (1.5pt);
\draw[color=qqqqff] (15.17,2.34) node {$V_{2n}$};
\fill [color=qqqqff] (3,9) circle (1.5pt);
\draw[color=qqqqff] (3.15,9.34) node {$C_i = (X_1 \vee X_2 \vee \overline{X}_{n-1})$};
\fill [color=qqqqff] (1,5) circle (1.5pt);
\draw[color=qqqqff] (1.14,5.33) node[left] {$C_{11}$};
\fill [color=qqqqff] (2,7) circle (1.5pt);
\draw[color=qqqqff] (2.14,7.34) node [left]{$C_{1\frac{n}{2}}$};
\fill [color=qqqqff] (2.59,6.98) circle (1.5pt);
\draw[color=qqqqff] (2.74,7.32) node [right]{$C_{2\frac{n}{2}}$};
\fill [color=qqqqff] (2.4,4.96) circle (1.5pt);
\draw[color=qqqqff] (2.53,5.29) node [right] {$C_{21}$};
\fill [color=qqqqff] (4.28,6.96) circle (1.5pt);
\draw[color=qqqqff] (4.41,7.3) node [right]{$C_{\overline{n-1}\frac{n}{2}}$};
\fill [color=qqqqff] (5.99,4.86) circle (1.5pt);
\draw[color=qqqqff] (6.14,5.2) node [right] {$C_{\overline{n-1}1}$};
\end{scriptsize}

\end{tikzpicture}

\caption{An example of graph $G(\phi)$ }
\end{figure}

\begin{theorem}
	 $G(\phi)$  is a $(\frac{n}{2}+1)$-laminar graph iff $\phi$ is satisfiable. 
\end{theorem}

\begin{proof}
	 
	 Suppose $\phi$ is satisfiable and let $\mathbb{A}$ be some satisfying truth assignment of the variables. Consider a path $\mu$  from $V_1$ to $V_{2n}$ forced to visit   the vertex $X_i$ if the variable is set to true in $\mathbb{A}$ and $\overline{X_i}$ otherwise. 

$\mu$ is obviously a diametral path. Since 	$\mathbb{A}$ is a truth assigment every clause $C_j$ of $\phi$ has a true literal which belongs to $\mu$
and therefore $d(C_j, \mu) \leq \frac{n}{2} + 1$.
All other vertices either belongs to $\mu$ are at distance 1.  Therefore $\mu$ is  ($\frac{n}{2} + 1$)-dominating diametral path.
	 
	Conversely, suppose $G(\phi)$ is $(\frac{n}{2}+1)$-laminar, hence using Lemma \ref{necessaire} there exists a diametral path $\mu$ of length $3n+2$ such that  a every vertex is at distance $\frac{n}{2}+1$ from $\mu$. As explained above $X_i$ and $\overline{X_i}$ can not be both on a diametral path. We set the variable $x_i$ to be true if $\mu$ passes through the vertex $X_i$ to false otherwise. Every clause $C_j$  must be satisfied because there is at least one variable vertex $X_j$ at distance $\frac{n}{2}+1$ from it.  Therefore this ($\frac{n}{2} + 1$)-dominating diametral path provides a truth assignment for $\phi$.
\end{proof}

It is obvious that the transformation can be computed in polynomial time.
Let us consider the following decision problem :

\noindent
\textbf{Name:} Laminarity

\noindent
\textbf{ Data:} A graph $G$ and $k$ an integer  such that $k \in \lbrack \frac{\sqrt{|V(G)|}}{4}, \frac{\sqrt{|V(G)|}}{2}\rbrack $ 

\noindent
\textbf{Question:} Is $G$ $k$-laminar ?

\begin{corollary}
Laminarity is an   NP-complete problem.
\end{corollary}

\begin{proof}

If we consider the 3SAT NP-complete variant in which every variable occurs at most 3 times  \cite{Papadimitriou}.
The relationship between the number of variables $n$  of such an instance and its number of clauses $m$ is :

$2n \leq m_{\phi} \leq 3n$  where $m_{\phi}$ denotes the total number of occurences of  variables in clauses. This inequalities just say that each variable has 2 or 3 occurences in the clauses, since we can get rid of the cases where a variable occurs only in one clause.

Considering the first inequalities we deduce:

$4n +2n(\frac{n}{2}+1) \leq  |V(G)| \leq 4n +3n(\frac{n}{2}+1)$, which gives :
$n^2 + 6n \leq  |V(G)| \leq 3\frac{n^2}{2} +7n$.

Replacing  $n$  by $2k -1$  we obtain : $4k^2 + 8k -5  \leq  |V(G)| \leq 12k^2 +2K -4$.

Therefore :
$4k^2   \leq  |V(G)| \leq 16k^2 $

If we consider the range $ \lbrack \frac{\sqrt{|V(G)|}}{4}, \frac{\sqrt{|V(G)|}}{2}\rbrack $  for $k$, using the construction described above we can encode all instances of a  NP-complete variant of 3SAT.
\end{proof}


 


\section{Conclusion and perspectives}
It would be interesting to improve  the running time for the recognition of k-laminar graphs (especially for 1-laminar ones).
But it should be noticed that for graphs having a constant number of extremal vertices (i.e., $|MaxEcc(G)|$ is bounded by this constant) then 
the complexity of the algorithms proposed here in theorems \ref{kstrongly},\ref{1-lam} goes down to $O(nm)$ which could be optimal, see \cite{BCH14, AWW16}. In particular when dealing with read networks their laminar parts seem to have a bounded number of extremal vertices.

One of the few theoretical results on clustering for restricted graph classes is presented in \cite{DKS97} and proposes  an approximation algorithm for 1-laminar graphs.
Therefore we think that these bio-inspired  k-laminar graphs are worth to be  studied further. As for example, searching for  diameter computations in linear time using a constant number of BFSs  as in \cite{BCHKMT15} and may have other applications not only in bioinformatics.


Perhaps the k-laminar class of graphs is too large to capture all properties of read networks.
The good notion could be k-diametral path graphs with its recursive definition for all induced subgraphs. Unfortunately there is no polynomial recognition algorithm for this class. A good algorithmic compromise would be to add some connectivity requirements, i.e., k-laminar and h-connected.
It would be interesting to develop a robust decomposition method  of read networks into their  k-laminar parts. In other words we want to  find a  skeleton of  a read network that captures most of its biological properties. Such a decomposition could provide an interesting alternative process to  analyze the biodiversity of  read networks.

\textbf{Acknowledgements:}
The authors wish to thank Anthony Herrel for many discussions on the project and for having selected the lizards on which this study is based.

\begin{small}

\end{small}
\newpage
\section{Appendix}
\subsection{k-laminar recognition algorithm}


We first notice that every graph $G$ is trivially $diam(G)$-laminar, and let us now generalize the previous recognition algorithm 1 to any fixed integer $k$ such that : $k < diam(G)$.

\begin{theorem}\label{recog}
 For every fixed $k\geq 2$ such that $k < diam(G)$, the algorithm k-Dominating-Diameter(G,s) finds a k-dominating diametral path starting form $s$ if some exists  in $O(n^{2k})$.
\end{theorem}

\begin{proof} 

To generalize Dominating-Diameter(G,s) algorithm, we will proceed similarly from a given vertex $s \in MaxEcc(G)$, by considering all the paths of length $k$ starting at $s$ and then make them grow layer by layer keeping only those which are potential extendable to a k-dominating diametral path.

We keep the same preprocessing as for the recognition of 1-laminar graphs, namely: we can compute $ecc(x)$ for every vertex $x$ of $G$. Afterwards $\forall s \in MaxEcc(G)$, we process a BFS and let us denote by  $T_s$ the associated BFS-tree.  $L_i$ represent the different layers of the BFS-tree, i.e. by convention $L_0=\{s\}$ and  $L_i$ is equal to the i-th neighborhood of $s$. Then $\forall v \in V(G)$, let us denote by $Level_s(v)$ its level in $T_s$.
We can also preprocess in the same time :
$\forall v \in V(G)$ and   $\forall i$ such that  $Level_s(v)$-$k \leq i \leq Level_s(v)+k$  we compute  $ N^k_i(v)=N^k(v) \cap L_i$. 
Since $k$ is fixed, the sets $N^k_i(v)$ can be computed in $O(nm)$ also.

\begin{figure}[ht!]
\begin{algorithm}[H]
\textbf{k-Dominating-Diameter(G,s):} 

\KwData{a graph $G=(V,E)$ and a start vertex $s \in MaxEcc(G)$, an integer $k \geq 2$\;}
\KwResult{YES / NO  $G$ has a k-dominating diametral path starting at $s$\;}

\textbf{Initialisations:}


Initialize  a queue \textit{Queue} with all different pairs $(v,P)$ such that $P$ is a path of length $k$ starting at $s$ in the BFS-tree, and $v$ its last vertex. This list is supposed to be lexicographically ordered accordingly to the layer orderings.


\While{$Queue \neq \emptyset$}{dequeue $(v, P)$ from beginning of  \textit{Queue}\;
 
$p \leftarrow Level_s(v)$ \;
$q\leftarrow p-k+2$ \;
$A(v) \leftarrow \bigcup_{u \in P, q-k+1 \leq Level_s(u) \leq p} N^k_{q}(u)$\;

\If{$p=diam(G)$}{\If {$\forall i$, $q \leq i \leq p$,  $L_i=\bigcup_{u \in P, i-k+1\leq Level_s(u) \leq i+k} N^k_{i}(u)$}
{ \textbf{YES} "a k-dominating diametral path  from $s$ to $v$ has been found"\textbf{STOP}}
}
\For{ $\forall x \in N_{p+1}(v)$}{

\If{$L_{q}$=$A(v) \cup N^k_{q}(x)$}{

$P' \leftarrow P+x$ \;
enqueue $(x,P')$ to the end of \textit{Queue} \;
			
}
}		
}
\textbf{NO}  "$G$ has no k-dominating diametral path starting at $s$"\;
\vspace{0.5cm}
\caption{}
\end{algorithm}
\end{figure}

\textbf{Invariant:} If the pair $(v,P)$ with $P=[s, \dots v]$ belongs to Queue, and if $p=Level_s(v)$ then $P$ a  path  k-dominating the first $p-k+1$ layers.

This invariant is clearly satisfied with the initializations of Queue. Then During the While loop a new pair $(v, P)$ is only inserted  if it satisfies this property.

\textbf{Complexity Analysis}:
The initialisation step may costs $O(n^k)$ since we could have $\prod_{i=1}^{h}|L_i|$ different pairs $(v, P)$.
The queue data structure forces the vertices to be visited in a breadth first way, giving an $O(n+m)$ to the managment of the while loop.
During this while loop:

For every vertex $v$ the set $A(v)$ is used at most $|N_{p+1}(v)|$ times, so in the whole it is bounded by $O(nm)$.
But to compute  the sets $A(v)$ we have to maintain paths of length 2k.
Unfortunately  there could be $n^{2k}$ such paths. This yields a polynomial algorithm in $O(n^{2k})$.

\end{proof}

\begin{corollary}
 k-laminar graphs can be recognized in $O(|MaxEcc(G)|.n^{2k})$ or $O(n^{2k+1})$. 
\end{corollary}
\begin{proof}
First  we have to compute all eccentricites in $G$ in $O(nm)$ and  then it is enough to repeat this Algorithm 2 for every $x \in MaxEcc(G)$, this provides an algorithm running $O(|MaxEcc(G)|.n^{2k})$.

\end{proof}

\end{document}